\documentclass[10pt, conference,letterpaper]{IEEEtran}
\usepackage{amsmath}
\usepackage{amssymb}
\usepackage{amsfonts}
\usepackage{amscd}
\usepackage{mathrsfs}
\usepackage{graphicx}
\usepackage[usenames,dvipsnames]{color}
\usepackage{bm}
\usepackage{url}
\usepackage{verbatim}
\usepackage{algpseudocode}

\newtheorem{definition}{Definition}
\newtheorem{proposition}{Proposition}

\newtheorem{example}{Example}

\newcommand{\FF}{\mathbb{F}}
\newcommand{\cv}{\mathbf{c}}
\newcommand{\xv}{\mathbf{x}}
\newcommand{\zv}{\mathbf{z}}

\begin{document}

\title{Sparsity Exploiting Erasure Coding for Resilient Storage and Efficient I/O Access in Delta based Versioning Systems}

\author{
\authorblockN{J. Harshan, Fr\'ed\'erique Oggier}
\authorblockA{School of Physical and Mathematical Sciences,\\
Nanyang Technological University, Singapore\\
Email:\{jharshan,frederique\}@ntu.edu.sg\\
}
\and
\authorblockN{Anwitaman Datta}
\authorblockA{School of Computer Engineering,\\
Nanyang Technological University, Singapore\\
Email:anwitaman@ntu.edu.sg\\
}
}

\maketitle

\begin{abstract}
In this paper we study the problem of storing reliably an archive of versioned data. Specifically, we focus on systems where the differences (deltas) between subsequent versions rather than the whole objects are stored - a typical model for storing versioned data. For reliability, we propose erasure encoding techniques that exploit the sparsity of information in the deltas while storing them reliably in a distributed back-end storage system, resulting in improved I/O read performance to retrieve the whole versioned archive. Along with the basic techniques, we propose a few optimization heuristics, and evaluate the techniques' efficacy analytically and with numerical simulations. 
\end{abstract}


%
%

\section{Introduction}

Using deltas is a well known technique to store a sequence of versions of a data object, where the differences between consecutive versions, rather than complete object instances themselves are maintained. It is used for a variety of applications, e.g.: (1) consider a user working on a local copy of his data, who explicitly saves/commits versions using a tool like the popular version management system, Subversion, a.k.a. SVN \cite{svn}. Then SVN is keeping the differences (`deltas') across consecutive versions, instead of all versions. 
(2) A very different kind of application is Wikipedia, which likewise keeps track of the differences between article contents, so that it is easy to track/revert changes, or identify vandalism. (3) Deltas are also exploited by cloud based back-up services, to reduce the network usage when uploading/downloading (synching) data, to give users old file versions from previous back-ups. 

As illustrated by the previous examples, the notion of differences (deltas) is particularly suited to the storage of multiple versions of the same data objects.

In this paper, we are interested in looking at the back-end storage systems to store the versioned data reliably. This reliability in the back-end system is derived by applying two mechanisms, (i) distribution - i.e., deployment of multiple storage devices, so that even if some of the devices fail, there are other storage devices which can still serve the data, (ii) storing the data redundantly over this distributed storage network, so that despite the loss of individual storage devices, enough information is retained in the system such that the original data is preserved. The redundancy can be obtained by replicating the data, or by employing erasure coding techniques, which are known (see e.g. \cite{google,hadoop}) to achieve better fault tolerance vis-a-vis storage overhead. This in turn has also led to a renewed interest in designing new erasure codes, aimed to address peculiarities of distributed storage systems \cite{OgD2}. Existing works are however predominantly geared towards storing immutable content, unlike the case of versioned data.
The recent works which do focus on mutable content do so in the context of efficiently carrying out an update \cite{RVBS,HPZV,MWC,ECD}, and thus focuses only on the storage of the latest version of the data. 

In contrast to these existing works, we address the question of efficient storage of versioned data and design a novel erasure coding framework - Sparsity Exploiting Coding (SEC) - where the version differences are erasure encoded, instead of encoding each version individually, and the sparsity of information across versions is opportunistically exploited to optimize the system's (disk) I/O performance during retrieval of the versioned archive. 
We evaluate the efficacies of the presented framework using static resiliency analysis, and do so by studying both systematic and non-systematic maximum distance separable (MDS)  codes (see Section \ref{sec:model} for a definition of MDS and systematic), and for different redundancy placement strategies. 
Our analysis demonstrates that the number of I/O accesses required is significantly reduced when retrieving the multiple versioned data archive. Due to a lack of consensus on proper workloads to study such systems (see e.g. \cite{dedup}), we experimented with a few example scenarios, where I/O reduction of up to 20\% were observed to retrieve a data object with 5 versions, while reductions between 4-13\% were observed for even just two versions, in randomized experiments where probability distributions on the sparsity of differences were chosen to study scenarios ranging from unfavorable to favorable to the proposed framework.

%
%

\section{System Model \& Preliminaries} 
\label{sec:model}

We start by providing a formal framework that describes erasure coding. 
Let $\FF_q$ denote the finite field with $q$ elements, where $q$ is a prime power, typically a power of 2 here. We will denote by $\xv \in \FF_q^k$ a data object to be stored over a storage network, that is, the data object is seen as a vector of $k$ blocks taking value in the alphabet $\FF_q$. We assume a fixed sized data object, and in particular that the modifications of this object do not change its length, which does not readily translate to application level objects such as files or directories. Thus, we implicitly assume that the application level objects are split and transformed into fixed sized objects (arguably with necessary zero padding), which is then used as the input $\xv \in \FF_q^k$ for the encoding process. The nuances of this transformation, as well as the subsequent reassembly of the whole files to be used by the applications is beyond the scope of this work, and all our subsequent discussions will instead be centered around the abstract data objects represented by $\xv \in \FF_q^k$. 

We consider the scenario of an erasure coding based distributed storage system, where fault tolerance is achieved using linear erasure codes. Recall that to archive an object $\xv \in \FF_q^k$, we first encode it using an $(n,k)$ linear code, that is $\xv$ is mapped to the codeword
\begin{equation}\label{eq:lincode}
\cv = \mathbf{G}\xv \in \FF_q^n,~n>k,
\end{equation}
for $\mathbf{G}$ an $n\times k$ matrix with coefficients in $\FF_q$ called {\em generator matrix} or sometimes coding matrix. The ratio $k/n$ is called the {\em rate} of the code.
We use the term {\em systematic} to refer to a codeword $\cv$ whose $k$ first components are $\xv$, that is $c_i=x_i$, $i=1,\ldots,k$. Once a codeword of length $n$ is obtained, all the $n$ coefficients $c_i$, $i=1,\ldots,n$ are stored across $n$ distinct nodes of the network. We say that an $(n,k)$ linear code is MDS (which stands for Maximum Distance Separable) when any patterns of $n-k$ failures can be tolerated.

Let $\xv_1 \in \FF_q^k$ be the first version of a data object to be stored. The data owner 
may at any time decide to modify it, giving rise to a new version of this data object, 
denoted by $\xv_2 \in \FF_q^k$. More generally, a new version $\mathbf{x}_{j+1}$ is obtained from $\xv_j$, and over time, we obtain a sequence $\{ \mathbf{x}_{j} \in \mathbb{F}^{k}_{q},~j=1,2,\ldots, L <\infty\}$ of different versions of a data object, to be stored in the network. The bit level-wise modifications between two successive versions are modelled by
\begin{equation}
\label{update_eq}
\mathbf{x}_{j + 1} = \mathbf{x}_{j} + \mathbf{z}_{j + 1},
\end{equation}
where $\mathbf{z}_{j + 1} \in \mathbb{F}^{k}_{q}$ keeps track of the changes in the $j$-th update. 

From a user point of view, the difference between two consecutive versions is determined by the application semantics. From a back-end storage system view however, we are interested in sequences $\{ \mathbf{x}_{j} \in \mathbb{F}^{k}_{q},~j=1,2,\ldots, L <\infty\}$ of data objects in their bit level representation, and exploit opportunistically the fact that often $\xv_{j+1}$ and $\xv_j$ may have little differences at the bit level or said differently $\zv_{j+1}$ in (\ref{update_eq}) is sparse (formally defined in Definition \ref{def:gamma}). As motivated in the introduction, version management systems like SVN \cite{svn} store differences (deltas) across versions, and are natural candidates to benefit from the proposed coding strategy.

\begin{definition}\label{def:gamma}
For some integer $1 \leq \gamma < k$, a vector $\mathbf{z} \in \mathbb{F}^{k}_{q}$ is said to be $\gamma$-sparse if it contains at most $\gamma$ non-zero entries. 
\end{definition} 

Once $\zv_{j+1}\in \mathbb{F}^{k}_{q}$ is $\gamma$-sparse, it suggests that it should be possible to access it more efficiently (with less I/O reads) than a normal data object. 
Indeed, the ideal case would be if one could only use $\gamma$ I/O reads, since the other $k-\gamma$ positions contain zeroes, without having to look for the non-zero positions.
We will show in next section, by proposing an explicit coding strategy, that it is possible to reduce the number of I/O reads from $k$ to $2\gamma$ I/O reads, which thus becomes beneficial when $\gamma < \frac{k}{2}$. The ideal case of $\gamma$ I/O reads is not achieved, since in practice we do not know the positions of the zeroes.   

%
%
\section{Sparsity Exploiting Coding (SEC)}
\label{sec2}

Let $\{\mathbf{x}_{j} \in \mathbb{F}^{k}_{q},~1\leq j \leq L\}$ be the sequence of versions of a data object to be stored in the network, where $\xv_j$ is the $j$th version (or version at the $j$-th instant of time). The number of components modified from $\mathbf{x}_{j}$ to $\mathbf{x}_{j + 1}$ is reflected in the vector $\mathbf{z}_{j + 1}=\xv_{j+1}-\xv_j$ in \eqref{update_eq} which is then $\gamma_{j+1}$-sparse (see Definition \ref{def:gamma}) for some $1 \leq \gamma_{j+1} \leq k$. 
We propose an encoding strategy using an $(n,k)$ linear erasure code (see (\ref{eq:lincode})) which exploits the sparsity of the differences across updates, thus referred to as {\em sparsity exploiting coding (SEC)}.
Note that the value $\gamma_{j+1}$ may a priori vary across updates of the same object and across different objects, and that sparsity is exploitable only when $\gamma_{j+1} < \tfrac{k}{2}$. 


\begin{figure}
\begin{algorithmic}[1]
\Procedure{Encode}{$\mathcal{X}, \mathbf{G}$}
\State \textbf{FOR} $0 \leq j \leq L-1$
\State $~~~$ \textbf{IF} $j = 0$
\State $~~~$ $~~~$ return $\mathbf{c}_{1} = \mathbf{G}\mathbf{x}_{1}$;
\State $~~~$ \textbf{ELSE} (\emph{This part summarizes Step $j+1$ in text})
\State $~~~$ $~~~$ Compute $\mathbf{z}_{j+1} = \mathbf{x}_{j+1} - \mathbf{x}_{j}$;
\State $~~~$ $~~~$ Compute and Store $\gamma_{j+1}$;
\State $~~~$ $~~~$ return $\mathbf{c}_{j+1} = \mathbf{G}\mathbf{z}_{j+1}$;
\State $~~~$\textbf{END IF}
\State \textbf{END FOR}
\EndProcedure
\end{algorithmic}
\caption{Encoding Procedure for SEC}\label{Algorithm1}
\end{figure}


\subsection{Object Encoding}
\label{sec2_subsec1}

The {\bf basic SEC method} to encode the $j$th version $\mathbf{x}_{j+1}$, $j\leq L-1$, using deltas is formally given by:

{\bf Step $j + 1$.} To encode the $(j+1)$-th version, the difference vector $$\mathbf{z}_{j+1} = \mathbf{x}_{j+1} - \mathbf{x}_{j}$$ and the corresponding sparsity level $\gamma_{j+1}$ are computed. Then the object $\mathbf{z}_{j+1}$ is encoded as either
$$\mathbf{c}_{j+1} = \mathbf{G}_S\mathbf{z}_{j+1},$$
if the coding matrix $\mathbf{G}_S \in \mathbf{F}^{n \times k}_{2}$ is in systematic form, or $$\mathbf{c}_{j+1} = \mathbf{G}_N\mathbf{z}_{j+1},$$
if $\mathbf{G}_N \in \mathbf{F}^{n \times k}_{2}$ is not in systematic form. 

The sparsity exploiting coding (SEC) procedure is summarized algorithmically in Figure \ref{Algorithm1}.
The input and the output of the algorithm are $\mathcal{X} = \{\mathbf{x}_{j + 1} \in \mathbb{F}^{k}_{q},~0\leq j \leq L-1\}$ and $\{\mathbf{c}_{j+1},~0\leq j \leq L-1\}$, respectively. 

The above description emphasizes the differential nature of the proposed SEC, where the first version is encoded in full while the subsequent versions are encoded via their subsequent differences. This leads to a recursive encoding of the object $\mathbf{x}_{l}$ for $l > 1$, whose overall storage pattern is $\{\mathbf{x}_{1}, \mathbf{z}_{2}, \ldots, \mathbf{z}_{L}\}$.

There are two main missing ingredients to complete the description of the proposed SEC: 
(1) explicit constructions for the coding matrices ${\bf G}_S$ and ${\bf G}_N$ that facilitate the recovery of $\mathbf{z}_{j+1}$ with fewer than $k$ I/O reads when $\gamma_{j+1} < \frac{k}{2}$, and (2) how the data placement of the objects $\{\mathbf{x}_{1}, \mathbf{z}_{2}, \ldots, \mathbf{z}_{L}\}$ should be done across the sets of nodes $\{\mathcal{N}_{1}, \mathcal{N}_{2}, \ldots, \mathcal{N}_{L}\}$, where the set $\mathcal{N}_{j+1}$ of nodes is used to store the components of $\mathbf{c}_{j+1}$.  
Both issues are equally important, and thus deserve a (sub)section of their own (see Subsection \ref{subsec:retrDEC} for the code design, and Section \ref{sec3} for the data allocation).

We conclude this subsection with some remarks, including two possible variations of the above SEC. These variations are not mutually exclusive and can be used in conjunction. 

{\bf Optimized Step $j + 1$.} A first variant of {\bf Step $j + 1$} is obtained by encoding a whole object if the sparsity level is too high, namely:
Store $\mathbf{c}_{j+1} = \mathbf{G}_S\mathbf{z}_{j+1}$ (or $\mathbf{c}_{j+1} = \mathbf{G}_N\mathbf{z}_{j+1}$) only when $\gamma_{j+1} < \frac{k}{2}$, and store $\mathbf{c}_{j+1} = \mathbf{G}_S\mathbf{x}_{j+1}$ (or $\mathbf{c}_{j+1} = \mathbf{G}_N\mathbf{z}_{j+1}$), otherwise. 

The I/O advantages of the {\bf Optimized Step $j+1$} will be discussed with an example in Subsection \ref{sec2:subsec:example}.

{\bf Reversed SEC.}\label{rem1}
For applications where the latest archived versions of the object are frequently accessed, a variant of the proposed SEC method could be employed where the order of storing the difference vectors is reversed as $\{\mathbf{z}_{2}, \mathbf{z}_{3}, \ldots, \mathbf{z}_{L}, \mathbf{x}_{L}\}$, so as to favor the latest version access.

Finally, note that the SEC stores only the deltas, yet there is an implicit assumption that $\xv_j$ is known, in order to compute its difference with $\xv_{j+1}$, $j=1,\ldots,L-1$. A practical way to satisfy this requirement is to cache a full copy of the latest version $\xv_j$, until a new version $\xv_{j+1}$ arrives. Keeping a cache of the latest version also helps in improving the response time and overheads of data read operations in general. Alternatively, the second variation above, Reversed SEC, can be applied, where the latest version is encoded fully, along with differences of older versions.

\subsection{Object Retrieval with Non-Systematic SEC}
\label{subsec:retrDEC}

Suppose that the $L$ versions of a data object have been archived, and the user needs to retrieve $\mathbf{x}_{l}$ for some $1 < l \leq L$. 
We discuss the procedure to retrieve $\mathbf{x}_{1}, \mathbf{z}_{2}, \ldots, \mathbf{z}_{l}$ from $\mathcal{N}_{1}, \mathcal{N}_{2}, \ldots, \mathcal{N}_{l}$. The recovery procedure depends on the structure of the SEC generator matrix, and hence, we explain the procedure considering two separate cases where the coding matrix is (i) $\mathbf{G}_{N}$ (non-systematic) and (ii) $\mathbf{G}_{S}$ (systematic). 
We start with the non-systematic case.

To retrieve $\mathbf{x}_{1}$, choose a subset of $k$ nodes from $\mathcal{N}_{1}$ to obtain 
\begin{equation*}
\mathbf{y} = \mathbf{G}_{\small{sub}}\mathbf{x}_{1},
\end{equation*}
where $\mathbf{G}_{\small{sub}} \in \mathbb{F}_{q}^{k \times k}$ is a submatrix of $\mathbf{G}_{N}$ which is invertible, then recover $\mathbf{x}_{1}$ as
\begin{equation*}
\mathbf{x}_{1} = \mathbf{G}^{-1}_{\small{sub}}\mathbf{y}.
\end{equation*}
We need to make sure that such a submatrix always exists, which gives us a first design criterion:
\begin{itemize}
\item {\bf Criterion 1.} There is at least one $k \times k$ submatrix of $\mathbf{G}_{N}$ that is full rank (to retrieve $\{\mathbf{x}_{1}, \mathbf{z}_{j} ~|~ \gamma_{j} \geq \frac{k}{2} \}$).
\end{itemize}

The retrieval procedure for $\{ \mathbf{z}_{j}, 2 \leq j \leq l \}$ depends on the corresponding sparsity levels $\{ \gamma_{j}, 2 \leq j \leq l \}$. If $\gamma_{j} \geq \frac{k}{2}$, then $\mathbf{z}_{j}$ is recovered using the same procedure as that of $\mathbf{x}_{1}$. If $\gamma_{j} < \frac{k}{2}$, choose a subset of $2\gamma_j$ nodes from $\mathcal{N}_{j}$ to obtain 
\begin{equation*}
\mathbf{y} = \mathbf{G}_{\gamma_{j}}\mathbf{z}_{j},
\end{equation*}
where $\mathbf{G}_{\gamma_{j}} \in \mathbb{F}^{2\gamma \times k}$ is a submatrix of $\mathbf{G}$.  
This gives us our second code design criterion, which follows from \cite{ZP}:
\begin{proposition}
\label{prop1}
If any $2\gamma$ columns of the $2\gamma \times k$ matrix $\Phi$ are linearly independent, then it is possible to uniquely recover the $\gamma$-sparse vector $\mathbf{z}$ from $\Phi\mathbf{z}$.
\end{proposition}
\begin{proof}
Since $2\gamma < k$, we can view $\Phi$ as the parity check matrix of a $(k, k - 2\gamma)$ linear code $\mathcal{C}$ in $\mathbb{F}^{k}_{q}$. For the matrix $\Phi$, if any $2\gamma$ columns of $\Phi$ are linearly independent, then the minimum Hamming distance of $\mathcal{C}$ is at least $2\gamma + 1$. Thus, from the properties of a linear code, $\mathcal{C}$ can correct all error patterns of weight less than or equal to $\gamma$, which in turn implies that it is possible to uniquely recover a $\gamma$-sparse vector $\mathbf{x}$.
\end{proof}

\begin{itemize}
\item {\bf Criterion 2.} For every $\gamma_{j} < \frac{k}{2}$, there is at least one $2\gamma_{j} \times k$ submatrix of $\mathbf{G}_{N}$ for which any $2\gamma_{j}$ columns are linearly independent (to retrieve $\{\mathbf{z}_{j} ~|~ \gamma_{j} < \frac{k}{2} \}$.)
\end{itemize}

It is clear that a minimum of $k$ I/O reads are needed to retrieve $\mathbf{x}_{1}$. However, to recover $\mathbf{z}_{j}$ for $2 \leq j \leq l$, the number of I/O reads is $\mbox{min} (2 \gamma_{j}, k)$. Overall, the total number of I/O reads to retrieve $\mathbf{x}_{l}$ in the differential set up is 
\begin{equation}
\label{reads_dec}
\eta(\mathbf{x}_{l}) = k + \sum_{j = l' + 1}^{l} \mbox{min} (2 \gamma_{j}, k),
\end{equation}
where $l' = 1$ for the basic encoding method ({\bf Step $j+1$}, $j\leq L-1$). For the optimized method ({\bf Optimized Step $j+1$}, $j\leq L-1$), $l' \leq l$ corresponds to the most recent version such that $\gamma_{l'} \geq \frac{k}{2}$. Finally, since the decoding method is differential, the procedure to read the first $l$ versions is the same as that for reading $\mathbf{x}_{l}$ for both the basic and the optimized method. Hence, the total number of I/O reads to retrieve the first $l$ versions is 
\begin{equation}
\label{eq:readall}
\eta(\mathbf{x}_{1}, \mathbf{x}_{2}, \ldots, \mathbf{x}_{l}) = k + \sum_{j = 2}^{l} \mbox{min} (2 \gamma_{j}, k).
\end{equation}

When there are node failures, different contenders for $\mathbf{G}_{sub}$ and $\mathbf{G}_{\gamma_{j}}$ give us the option to retrieve the objects before the node repair process. Hence, it is beneficial for the overall system performance to relax the condition of \emph{at least one submatrix} to \emph{several submatrices} in the criteria 1 and 2.


\begin{example}
\label{example_non_sys}
Consider an $(n, k)$ maximum distance separable (MDS) code whose generator matrix $\mathbf{G}_{N}$ is given by the Cauchy matrix
\begin{equation}
\label{cauchy_matrix}
\mathbf{G}_{N} = \left[\begin{array}{rrrrrrr}
g_{1,1} & g_{1,2} & \ldots & g_{1,k}\\
g_{2,1} & g_{2,2} & \ldots & g_{2,k}\\
\vdots & \vdots & \vdots & \vdots\\
g_{n,1} & g_{n,2} & \ldots & g_{n,k}\\
\end{array}\right],
\end{equation}
where $g_{i,j} = \frac{1}{h_{i} - f_{j}}$ for $\{ h_{i} \in \mathbb{F}_{q}, 1 \leq i \leq n \}$ and $\{ f_{j} \in \mathbb{F}_{q}, 1 \leq j \leq k\}$ such that $h_{i} - f_{j} \neq 0$ $\forall ~i, j$. Since any square submatrix of a Cauchy matrix is full rank over a finite field \cite{LaF}, any $2\gamma_{j} \times k$ submatrix  of $\mathbf{G}_{N}$ satisfies Proposition \ref{prop1}. Thus, MDS codes from Cauchy matrices are readily applicable in the proposed differential set up.
\end{example}

\subsection{Object Retrieval with Systematic SEC}
\label{sec3:subsec:sys_retr}

We next consider the case where the $L$ versions of a data object are archived using a systematic code, and the user needs to retrieve $\mathbf{x}_{l}$ for some $1 < l \leq L$.
The generator matrix of the systematic code is of the form
\begin{equation*}
\mathbf{G}_{S} =  \left[\begin{array}{c}
\mathbf{I}_{k}\\
\mathbf{B}\\
\end{array}\right],
\end{equation*}
where $\mathbf{I}_{k}$ is the $k \times k$ identity matrix and $\mathbf{B} \in \mathbb{F}^{n-k \times k}_{q}$ generates the $n-k$ parity symbols. Since the code is systematic, the objects $\mathbf{x}_{1}$ and $\{\mathbf{z}_{j}, \forall \gamma_{j} \geq \frac{k}{2}\}$ can be retrieved by downloading the contents from the $k$ systematic nodes. If $\gamma_{j} < \frac{k}{2}$, then choose a subset of $2\gamma_j$ nodes from $\mathcal{N}_{j}$ to obtain 
\begin{equation*}
\mathbf{y} = \mathbf{G}_{\gamma_{j}}\mathbf{z}_{j},
\end{equation*}
where $\mathbf{G}_{\gamma_{j}} \in \mathbb{F}^{2\gamma \times k}$ is a submatrix of $\mathbf{G}_{S}$. If $\mathbf{G}_{\gamma_{j}}$ satisfies Proposition \ref{prop1}, then $\mathbf{z}_{j}$ can be recovered. Note that the submatrix satisfying {\bf Criterion 2} is most likely to come from $\mathbf{B}$. Indeed, suppose that any row of $\mathbf{I}_{k}$ is taken, then since its length $k$ satisfies $k>2\gamma$, any pattern of consecutive $2\gamma$ zeroes results in a $2\gamma \times 2\gamma$ submatrix which is not full rank, which is likely to happen whenever $2\gamma << k$. Restricting to the matrix ${\bf B}$ which has only $n-k$ rows leads to the constraint that $\gamma_{j}$-sparse updates can be recovered with $2\gamma_{j}$ I/O reads only for $2\gamma_{j}<n-k$, that is $\gamma_{j} < \frac{n-k}{2}$. The number of I/O reads to retrieve $\mathbf{z}_{j}$ with a systematic code whose matrix $B$ satisfies {\bf Criterion 2} is 
\begin{eqnarray*}
\eta_{j} = \left\{ \begin{array}{ccccc}
2\gamma_{j}, \mbox{ if } \gamma_{j} \leq \frac{n-k}{2} \\
k, \mbox{otherwise}.
\end{array} 
\right.
\end{eqnarray*}
Since we are interested in $\gamma_j<\frac{k}{2}$, either we have $\frac{n-k}{2} < \frac{k}{2} \iff \frac{k}{n} > \frac{1}{2}$, which implies that systematic erasure coding can only recover less than $\lceil \frac{k}{2} \rceil - 1$ sparse levels with reduced I/O, or, if $\frac{n-k}{2} \geq \frac{k}{2} \iff \frac{k}{n} \leq \frac{1}{2}$, and it can recover up to $\lceil \frac{k}{2} \rceil - 1$ sparse levels (which is the same as that of non-systematic encoding). The total number of I/O reads to retrieve $\mathbf{x}_{l}$ in the differential set up is 
$$\eta(\mathbf{x}_{l}) = k + \sum_{j = l' + 1}^{l} \eta_{j},$$ 
where $l' = 1$ for basic encoding. For the optimized method, $l' \leq l$ corresponds to the most recent version such that $\gamma_{l'} \geq \frac{k}{2}$. The total number of I/O reads to retrieve the first $l$ versions is also 
$$\eta(\mathbf{x}_{1}, \mathbf{x}_{2}, \ldots, \mathbf{x}_{l}) = k + \sum_{j = 2}^{l} \eta_{j}.$$


\begin{example}
Similarly to Example \ref{example_non_sys}, an $(n-k) \times k$ Cauchy matrix can be used to construct the matrix $\mathbf{B}$, resulting in an MDS code. 
\end{example}


\subsection{I/O Benefits: An Illustrative Example} 
\label{sec2:subsec:example}

Consider a differential storage system that stores $L = 5$ versions of an object of size $k = 10$ using a $(20, 10)$ erasure code that satisfies the desired design criteria. Let the sparsity levels of subsequent versions be $\{\gamma_{j} ~|~ 2 \leq j \leq L\} = \{3, 8, 3, 6\}$.  We also assume that $\mathbf{x}_{1}$ itself is not sparse. 
We compute the number $\eta(\xv_l)$ of I/O reads needed to retrieve the $l$th version $\xv_l$. Since the employed code has rate $\frac{k}{n}=\frac{1}{2}$, the below given I/O read numbers are applicable for both systematic and non-systematic cases. 

\textbf{Basic encoding (Step $j+1$, $j\leq L-1$).} In this technique, the stored objects are $\{\mathbf{x}_{1}, \mathbf{z}_{2}, \mathbf{z}_{3}, \mathbf{z}_{4}, \mathbf{z}_{5}\}$. The number of I/O reads to retrieve $\mathbf{x}_{1}, \mathbf{z}_{2}, \mathbf{z}_{3}, \mathbf{z}_{4}, \mathbf{z}_{5}$ are $10, 6, 10, 6,  10$, respectively. Thus, $\{\eta(\mathbf{x}_{l}), 1 \leq l \leq 5 \}$ is $\{10,  16,  26,  32,  42\}$. The total I/O reads to recover all the $5$ versions is 42 (instead of 50 for the non-differential method). Thus, there is a reduction in the number of I/O reads to retrieve all the versions.  

\textbf{Optimized encoding (Optimized Step $j+1$, $j\leq L-1$).} In this method, the stored objects are $\{\mathbf{x}_{1}, \mathbf{z}_{2}, \mathbf{x}_{3}, \mathbf{z}_{4}, \mathbf{x}_{5}\}$. The number of I/O reads to retrieve $\mathbf{x}_{1}, \mathbf{z}_{2}, \mathbf{x}_{3}, \mathbf{z}_{4}, \mathbf{x}_{5}$ are $10, 6, 10, 6,  10$, respectively. Thus, $\{\eta(\mathbf{x}_{l}), 1 \leq l \leq 5 \}$ is $\{10, 16, 10, 16,  10\}$. However, the number of I/O reads to recover all the $5$ versions is same as the basic encoding method. Note that the number of I/O reads to retrieve individual versions is lower than the basic encoding method. 




\section{Static Resilience Analysis}
\label{sec3}

We next compute the static resilience (the amount of failures that the system tolerates based on the initial redundancy, if no further remedial actions are taken) of the proposed SEC coding strategies that exploit the sparsity across subsequent versions. 
We suppose that $L$ versions of a data object are stored, namely the pattern of stored data is $\{\mathbf{x}_{1}, \mathbf{z}_{2}, \mathbf{z}_{3}, \mathbf{z}_{L-1}, \mathbf{z}_{L}\}$, and encoded pieces for any of these versions are stored in $n$ nodes. 
The static resilience is computed for two practical redundancy placement choices:
a (\textbf{dispersed placement}), where differences and first version are all stored in different nodes, involving a total of $nL$ distinct nodes, and a (\textbf{colocated placement}), when all the versions' encoded pieces are stored in a common set of $n$ nodes.

We use the non-differential strategy where each version is coded and stored individually as a baseline, and demonstrate that for both placements, both the systematic and non-systematic SEC schemes achieve the same resiliency as non-differential coding for a given overall storage overhead (i.e., there is no resiliency compromise), even though both new coding technique results in a reduction in access I/O when retrieving all the previous versions of the data.  

When comparing systematic and non-systematic strategies, a key difference is that the number of submatrices of $\mathbf{G}_{S}$ satisfying Criterion 2 is fewer compared to that of $\mathbf{G}_{N}$ (as explained in Section \ref{sec3:subsec:sys_retr}). We will demonstrate that this reduction in the number of options results in a poorer resilience of individual version differences for the systematic SEC, compared to its non-systematic counterpart. However, the best resilience, taking into account all the versions of the data, is realized, for each of the two coding strategies, when the encoded pieces for each of the versions (or differences) are stored in the same set of nodes - which ultimately leads to the same net resilience for all the coding strategies.  

In the following, we assume that individual nodes fail with a probability $p$ and the failure events are independent.

\subsection{SEC Analysis} 

We start by computing the probability of losing one individual version among 
$\{\mathbf{x}_{1}, \mathbf{z}_{2}, \mathbf{z}_{3}, \mathbf{z}_{L-1}, \mathbf{z}_{L}\}$.
We assume SEC are MDS, since Cauchy matrices based SEC are MDS.

For a MDS SEC, whether it is systematic or not, any $k$ nodes storing the encoded pieces of $\mathbf{x}_{1}$ are sufficient to retrieve it. Thus $\mathbf{x}_{1}$ is lost if the event 
\begin{eqnarray*}
\mathcal{E}_{1} & = & \{ n - k + 1 \mbox{ or more nodes fail}\}
\end{eqnarray*}
occurs. The probability of losing $\mathbf{x}_{1}$ is then given by 
\begin{equation}\label{de1_nonsys}
\mbox{Prob}_{N}(\mathcal{E}_{1}) = 
\mbox{Prob}_{S}(\mathcal{E}_{1}) =  
\sum_{j=0}^{k-1}C^{n}_{n-j} p^{n-j}(1-p)^{j}.
\end{equation}
For any other arbitrary version $l$ ($2 \leq l \leq L$) in the {\bf non-systematic case}, where only the difference with its previous version $\mathbf{z}_{l}$ is stored, any $\upsilon_{l} = \mbox{min}(2\gamma_{l}, k)$ nodes suffice to retrieve $\mathbf{z}_{l}$ since any $\upsilon_{l} \times k$ submatrix of $\mathbf{G}_{N}$ satisfies Criterion 2. As a result, the object $\mathbf{z}_{l}$ is lost if the event 
\begin{eqnarray*}
\mathcal{E}_{l} & = & \{n - \upsilon_{l} + 1 \mbox{ or more nodes fail}\}
\end{eqnarray*}
occurs. The probability of losing $\mathbf{z}_{l}$ is given by
\begin{equation}
\label{de2_nonsys}
\mbox{Prob}_{N}(\mathcal{E}_{l})  =  
\sum_{j=0}^{\upsilon_l-1}C^{n}_{n-j} p^{n-j}(1-p)^{j}.
\end{equation}

Similarly, in the {\bf systematic case}, for $\mathbf{z}_{l}$ with $\gamma_{l} \geq \frac{k}{2}$, we have
\begin{eqnarray}
\mbox{Prob}_{S}(\mathcal{E}_{l}) & = & \mbox{Prob}_{S}(\mathcal{E}_{1}).
\end{eqnarray}
However, for $\mathbf{z}_{l}$ with $l < \frac{k}{2}$, not all combinations of $2\gamma_{l}$ nodes can retrieve $\mathbf{z}_{l}$ since only few $2\gamma_{l} \times k$ submatrices of $\mathbf{G}_{N}$ satisfies Criterion $2$. As a result, the probability of losing $\mathbf{z}_{l}$ is strictly lower bounded as
\begin{eqnarray}
\label{de2_sys}
\mbox{Prob}_{S}(\mathcal{E}_{l}) & > & 
\sum_{j=0}^{2\gamma_l-1}C^{n}_{n-j} p^{n-j}(1-p)^{j}.
\end{eqnarray}
Thus, we have
\begin{equation}
\label{ineq_sys_non_sys}
\mbox{Prob}_{S}(\mathcal{E}_{l}) \geq \mbox{Prob}_{N}(\mathcal{E}_{l}) \mbox{ for } 1 \leq l \leq L,
\end{equation}
where the inequality holds when $\gamma_{l} < \frac{k}{2}$, and equality otherwise.

{\bf Dispersed placement:} In a dispersed placement, the probability of retaining all the $L$ versions of the object are 
\begin{equation}
\label{eq:x1x2_avail_dist}
P_{d}(\mathbf{x}_{1}, \mathbf{x}_{2}, \ldots, \mathbf{x}_{L}) =  \prod_{l = 1}^{L} (1 - \mbox{Prob}(\mathcal{E}_{l})).
\end{equation}

Using the inequality of \eqref{ineq_sys_non_sys} in \eqref{eq:x1x2_avail_dist}, it is clear that non-systematic SEC provides at least same level of resilience as that of systematic codes if dispersed placement is employed. 

{\bf Colocated placement:} In a colocated placement, the probability of retaining all the $L$ versions of the object are 
\begin{eqnarray}
\label{eq:x1x2_avail_colc}
P_{c}(\mathbf{x}_{1}, \mathbf{x}_{2}, \cdots, \mathbf{x}_{L}) & = & 1 - \mbox{Prob}(\cup_{l = 1}^{L} \mathcal{E}_{l}),
\end{eqnarray}

For colocated placement, all the $L$ objects $\{ \mathbf{x}_{1}, \mathbf{z}_{2}, \cdots, \mathbf{z}_{L}\}$ can be recovered with probability one if any $k$ nodes are alive. Although some objects $\{ \mathbf{z}_{l} ~|~ \gamma_{l} < \frac{k}{2}\}$ can be retrieved even with the loss of $n - 2\gamma_{l} > n - k$ nodes, such failure patterns nevertheless result in the loss of $\mathbf{x}_{1}$, thereby making the recovery of all the $L$ versions impossible. Hence, the existence of any $k$ live nodes guarantees the recovery of $\{ \mathbf{x}_{1}, \mathbf{z}_{2}, \cdots, \mathbf{z}_{L}\}$ and this argument is applicable for both systematic and non-systematic SEC. With that, the probability of retaining all the $L$ versions is same for both systematic and non-systematic SEC and is given by
\begin{eqnarray}
\label{eq:colc_sys_nonsys}
P_{c}(\mathbf{x}_{1}, \mathbf{x}_{2}, \cdots, \mathbf{x}_{L}) & = & 1 - \mbox{Prob}_{S}(\mathcal{E}_{1}).
\end{eqnarray}
Note that if any $2\gamma_{l}$ nodes are sufficient to recover the sparse updates for the the non-systematic erasure codes, only specific patterns of $2\gamma_{l}$ nodes are applicable for the systematic codes. More discussion on the effects of different possible options to recover sparse updates are discussed in Section \ref{sec:num;results}.

Finally, comparing \eqref{eq:colc_sys_nonsys} and \eqref{eq:x1x2_avail_dist}, we conclude that colocated placement yields higher resilience for both systematic and non-systematic erasure codes than the dispersed placement. Henceforth, the resilience values calculated in colocated placement are used as the resilience of the systematic and non-systematic SEC methods. 

\subsection{Non-differential Coding (Baseline)} 

Any $k$ nodes of the $n$ nodes where encoded pieces of a version $\mathbf{x}_{l}$ for $1 \leq l \leq L$ are stored, is adequate to retrieve that version, i.e., $\mbox{Prob}_{ND}(\mathcal{E}_{1})$ = $\mbox{Prob}_{N}(\mathcal{E}_{1})$ in \eqref{de1_nonsys}.

Thus, for dispersed placement, probability of retaining\\ $\{\mathbf{x}_{1}, \mathbf{x}_{2}, \ldots, \mathbf{x}_{L} \}$ is
\begin{eqnarray}
\label{non_diff_dist}
P_{d}(\mathbf{x}_{1}, \mathbf{x}_{2}, \ldots, \mathbf{x}_{L}) = (1 - \mbox{Prob}_{ND}(\mathcal{E}_{1}))^{L},
\end{eqnarray}
For $l \geq 2$, we have $\mbox{Prob}_{ND}(\mathcal{E}_{1}) \geq \mbox{Prob}_{N}(\mathcal{E}_{l})$, where the inequality holds when $\gamma_{l} < \frac{k}{2}$, and equality otherwise. Using the above inequality in \eqref{non_diff_dist}, it is clear that non-differential erasure codes provide at most as much resilience as that of differential non-systematic SEC in dispersed placement.

For the case of collocated placement while using non-differential coding, 
\begin{eqnarray}
\label{non_diff_colc}
P_{c}(\mathbf{x}_{1}, \mathbf{x}_{2}, \ldots, \mathbf{x}_{L}) & = & (1 - \mbox{Prob}_{ND}(\mathcal{E}_{1})).\\
& \geq & P_{d}(\mathbf{x}_{1}, \mathbf{x}_{2}, \ldots, \mathbf{x}_{L})
\end{eqnarray}
Therefore, non-differential erasure codes have same resilience as that of the differential SEC codes, when encoded information about all the versions are colocated. Similar to the SEC methods, colocated placement is also optimal for the non-differential encoding method.

\subsection{Example}
\label{running_ex1}

We revisit the calculations above with a concrete example with specific parameter choices. This helps us to obtain specific values and compare the static resilience of the different schemes explicitly.

\begin{paragraph}{Set-up}
Consider a system that stores $2$ versions of a data object. Let the original data object be a binary file of size 3KB. We represent this object as a $3$-length vector $\mathbf{x}_{1}$ over the finite field of size $1\mbox{KB}$, i.e., let $\mathbf{x}_{1} \in \mathbb{F}^{3}_{q}$ where $q = 1024$. Further, let the second version of the object be such that only the first $1$KB of the binary file has been modified. It is important to note that the quantum and location of changes are not known a priori, and the coding scheme, decided in advance, has to work irrespective of the specificities of the update. In the finite field level, the second version is represented as $\mathbf{x}_{2} = \mathbf{x}_{1} + \mathbf{z}_{2}$, where $\mathbf{z}_{2}$ is $1$-sparse given by 
\begin{equation*}
\mathbf{z}_{2} =  \left[\begin{array}{c}
X\\
0\\
0\\
\end{array}\right] \in \mathbb{F}^{3}_{q},
\end{equation*}
where $X$ denotes a non-zero element of $\mathbb{F}_{q}$. 
\end{paragraph}

\begin{figure}
\begin{center}
\includegraphics[width=3in]{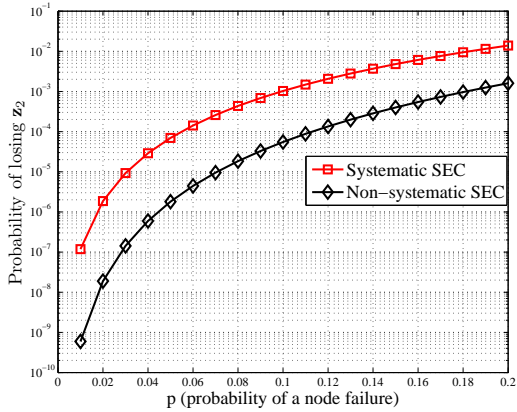}
\vspace{-0.5cm}
\caption{Probability of losing the $1$-sparse difference object $\mathbf{z}_{2}$ in example of Subsection \ref{running_ex1}: $\mbox{Prob}_{S}(\mathcal{E}_{2})$ in \eqref{pe_z2_sys} and $\mbox{Prob}_{N}(\mathcal{E}_{2})$ in \eqref{pe_z2_nonsys} for the systematic and non-systematic SEC respectively.}
\label{1sparseloss}
\end{center}
\end{figure}

\begin{paragraph}{The non-systematic case}
For the system parameters of our set-up, we pick a $(6, 3)$ non-systematic MDS code whose generator matrix $\mathbf{G}_{N} \in \mathbf{F}^{6 \times 3}_{q}$ is carved from a Cauchy matrix. Subsequently, $\mathbf{z}_{2}$ is encoded using $\mathbf{G}_{N}$ to obtain the codeword $\mathbf{c}_{2} = \mathbf{G}_{N}\mathbf{z}_{2} \in \mathbb{F}^{6}_{q}$. It is clear that the number of I/O reads needed to retrieve the first $2$ versions is $5$. Since the code is MDS, the probability of losing $\mathbf{x}_{1}$ is given by 
\begin{eqnarray}
\mbox{Prob}_{N}(\mathcal{E}_{1}) & = & p^6 + C^{6}_{5} p^5(1-p) + C^{6}_{4} p^4(1-p)^{2}.
\end{eqnarray}
Similarly, the probability of losing $\mathbf{z}_{2}$ is given by
\begin{eqnarray}
\label{pe_z2_nonsys}
\mbox{Prob}_{N}(\mathcal{E}_{2}) & = & p^6 + C^{6}_{5} p^5(1-p),\\
& < & \mbox{Prob}_{N}(\mathcal{E}_{1}) \nonumber.
\end{eqnarray}
The proposed non-systematic SEC opportunistically exploits sparsity in $\mathbf{z}_{2}$ to provide higher resilience for the object $\mathbf{z}_{2}$ than the object $\mathbf{x}_{1}$. However, in colocated placement, the resilience for both the objects $\mathbf{x}_{1}$ and $\mathbf{x}_{2}$ is dominated by that for $\mathbf{x}_{1}$. Hence, we have
\begin{eqnarray*}
P_{c}(\mathbf{x}_{1}, \mathbf{x}_{2}) & = & 1 - \left(p^6 + C^{6}_{5} p^5(1-p) + C^{6}_{4} p^4(1-p)^{2}\right).
\end{eqnarray*}
\end{paragraph}
\begin{center}
\begin{table*}
\begin{small}
\caption{Differential vs. Non-differential Erasure coding (numbers are based on Example of Subsection \ref{running_ex1})}
\begin{center}
\begin{tabular}{|c|c|c|c|c|c|c|c|c|c|c|}
\hline Version & Parameter & Differential & Differential & Non-differential \\
& & Non-systematic & Systematic & Systematic\\
\hline 1st & Encoding & $\mathbf{c}_{1} = \mathbf{G}_{N}\mathbf{x}_{1}$ & $\mathbf{c}_{1} = \mathbf{G}_{S}\mathbf{x}_{1}$ & $\mathbf{c}_{1} = \mathbf{G}_{S}\mathbf{x}_{1}$\\
\hline  & Encoding Complexity & matrix multiplication & matrix multiplication & matrix multiplication\\
& & & for parity only & for parity only\\
\hline & Nr. of nodes & 6 & 6 & 6\\
\hline & Decoding Complexity & inverse operation & low & low\\
\hline & I/O reads & 3 & 3 & 3\\
\hline 2nd & Encoding & $\mathbf{c}_{2} = \mathbf{G}_{N}\mathbf{z}_{2}$ & $\mathbf{c}_{2} = \mathbf{G}_{S}\mathbf{z}_{2}$ & $\mathbf{c}_{2} = \mathbf{G}_{S}\mathbf{x}_{2}$\\
&  & & & \\
\hline & Encoding Complexity & matrix multiplication & matrix multiplication & matrix multiplication\\
& &  & for parity only & for parity only\\
\hline & Nr. of nodes & 6 & 6 & 6\\
\hline & Decoding Complexity & sparse reconstruction & sparse reconstruction & low\\
\hline & I/O reads & 2 & 2 & 3\\
\hline
\end{tabular}
\end{center}
\label{table_rep_vs_er}
\end{small}
\end{table*}
\end{center}

\begin{paragraph}{The systematic case}
Now let us pick a $(6, 3)$ systematic erasure code whose generator matrix is given by $$\mathbf{G}_{S} = [\mathbf{I}_{3} ~\mathbf{B}^{T}]^{T} \in \mathbf{F}^{6 \times 3}_{q}$$ where $\mathbf{B} \in \mathbf{F}^{3 \times 3}_{q}$ is a Cauchy matrix. We encode the first version $\mathbf{x}_{1} \in \mathbb{F}^{3}_{q}$ as  
\begin{equation*}
\mathbf{c}_{1} = \mathbf{G}_{S} \mathbf{x}_{1} = \left[\begin{array}{c}
\mathbf{x}_{1}\\
\mathbf{B}\mathbf{x}_{1}\\
\end{array}\right] \in \mathbb{F}^{6}_{q}.
\end{equation*}
The storage overhead for the first version is $\frac{n}{k} = 2$. 
Subsequently, $\mathbf{z}_{2}$ is encoded using $\mathbf{G}_{S}$ to obtain the codeword 
\begin{equation*}
\mathbf{c}_{2} = \mathbf{G}_{S}\mathbf{z}_{2} = \left[\begin{array}{c}
\mathbf{z}_{2}\\
\mathbf{B}\mathbf{z}_{2}\\
\end{array}\right]\in \mathbb{F}^{6}_{q}.
\end{equation*}
For this scheme too, the number of I/O reads needed to retrieve the first $2$ versions is $5$.

Since the code is MDS, probability of losing $\mathbf{x}_{1}$ is given by
\begin{eqnarray}
\label{pe_x1_sys}
\mbox{Prob}_{S}(\mathcal{E}_{1}) & = & p^6 + C^{6}_{5} p^5(1-p) + C^{6}_{4} p^4(1-p)^{2}.
\end{eqnarray}
The difference object $\mathbf{z}_{2}$ is lost if $5$ or more nodes fail. In addition, since not all $2 \times 3$ submatrices of $\mathbf{G}_{s}$ satisfy Criterion 2, there are some specific 4 node failure patterns that lead to the loss of the object. Considering all possibilities, we have 
\begin{eqnarray*}
\mathcal{E}_{2} & = & \{\mbox{5 or more nodes fail}\} \cup \{\mbox{specific 4 nodes failure}\}.
\end{eqnarray*}
Hence, the probability of losing $\mathbf{z}_{2}$ is given by
\begin{eqnarray}
\label{pe_z2_sys}
\mbox{Prob}_{S}(\mathcal{E}_{2}) & = & p^6 + C^{6}_{5} p^5(1-p) + 12 p^4(1-p)^2,\\
& < & \mbox{Prob}_{S}(\mathcal{E}_{1}).\nonumber
\end{eqnarray}

In Fig. \ref{1sparseloss}, we compare $\mbox{Prob}_{S}(\mathcal{E}_{2})$ (for systematic SEC) and $\mbox{Prob}_{N}(\mathcal{E}_{2})$ (for non-systematic SEC) for different values of $p$. The plots show that systematic SEC, while exploiting the sparsity to reduce the I/O reads, does not provide higher protection for the difference object $\mathbf{z}_{2}$ when compared to the non-systematic SEC. Nevertheless, in colocated placement, the resilience for both the objects $\mathbf{x}_{1}$ and $\mathbf{x}_{2}$ is dominated by that for $\mathbf{x}_{1}$. Hence, we have
\begin{eqnarray*}
P_{c}(\mathbf{x}_{1}, \mathbf{x}_{2}) & = & 1 - \left(p^6 + C^{6}_{5} p^5(1-p) + C^{6}_{4} p^4(1-p)^{2}\right).
\end{eqnarray*}
From the point of view of failure events for $\mathbf{z}_{2}$, there are a total of $63$ possible failure patterns of nodes. Among them, both non-systematic and systematic methods can recover from 41 patterns due to the inherent MDS property, wherein, sparsity is not exploited and $\mathbf{z}_{2}$ is retrieved with 3 I/O reads. In addition to these failures, the non-systematic code can resist 15 more failure patterns arising from failure of all possible 4 node combinations. However, for such a case, the systematic SEC can resist only in additional 3 cases. Therefore, non-systematic SEC can handle a total of 56 failure patterns, while systematic can handle only 44 patterns. Thus, non-systematic SEC can opportunistically improve the resilience for the difference object $\mathbf{z}_{2}$ compared to systematic SEC.
\end{paragraph}

\begin{figure}
\begin{center}
\includegraphics[width=3in]{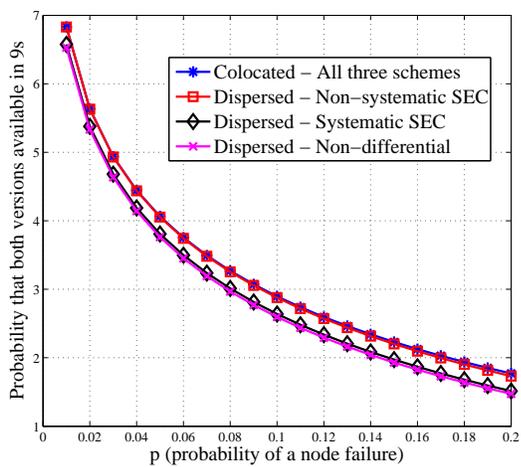}
\vspace{-0.5cm}
\caption{Resilience of colocated and distributed placement strategies for the example of Subsection \ref{running_ex1}. On the $y$-axis, the probability of joint availability of $\mathbf{x}_{1}$ and $\mathbf{x}_{2}$ in 9s format is shown, defined as $-\mbox{log}_{10}(1 - p_{d}(\mathbf{x}_{1}, \mathbf{x}_{2}))$ resp. as $-\mbox{log}_{10}(1 - p_{c}(\mathbf{x}_{1}, \mathbf{x}_{2}))$ for dispersed, resp. colocated placement.}
\label{wholeobjresplot}
\end{center}
\end{figure}

In Fig. \ref{wholeobjresplot} we show the probability of availability of all the versions (i.e., both $\mathbf{x}_{1}$ and $\mathbf{x}_{2}$ in Example of Subsection \ref{running_ex1}) of the data objects  for both dispersed (in \eqref{eq:x1x2_avail_dist}) and colocated placements (in \eqref{eq:x1x2_avail_colc}) for a range of values for the probability of failure $p$ of individual storage nodes. 
For colocated placement, all the three schemes have the same resilience to store $\mathbf{x}_{1}$ and $\mathbf{x}_{2}$. However, for dispersed placement non-systematic SEC provides higher resilience than the other two schemes.

\subsection{Resilience analysis summary}
A comprehensive summary of the resilience of the three schemes is provided in Table \ref{table_rep_vs_er}. The highlights are as follows:

(1) For any given choice of coding scheme, colocated placement of encoded pieces of multiple versions of a data object results in higher resilience than the dispersed placement of these encoded pieces. However, going into more subtlety, there is one advantage of using non-differential coding with dispersed storage, namely, some random versions (but not the whole versioned archive) will survive with a greater probability than the probability of survival of the whole archive for the colocated case. In contrast, since the basic SEC stores differences, this is not the case. The optimized SEC however benefits from this serendipity as well, but less often than for the non-differential case.

(2) For colocated placement, the non-systematic SEC provides the same resilience to retrieve the whole archive as that of the systematic SEC. However, larger number of options to recover the sparse updates for the former method results in higher resilience for storing the individual difference objects than the latter method. In addition, whenever $\frac{k}{n} > \frac{1}{2}$, non-systematic SEC works for a larger range of sparseness levels with reduced I/O than the systematic SEC.

(3) For colocated placement, the non-systematic SEC provides the same resilience as that of the non-differential method, but with the advantage of requiring fewer I/O reads than the latter when retrieving the whole versioned archive. 

Finally, we note that for the non-systematic SEC, the individual version deltas have higher static resilience (Fig \ref{1sparseloss}), however, there is no advantage in this, given that the actual retrievability is bottlenecked by the facts that (i) colocation is best strategy, and (2) the availability of the first version (where the whole object is coded) thus dominates and determines retrievability of the whole archive. This suggests that the additional resilience of individual deltas in the non-systematic SEC is wasteful in terms of storage resources, and that there is potential room to reduce storage overhead while storing with non-systematic SEC. Study of storage optimization for the non-systematic SEC will be part of our future work.

%
%

\section{Versioned Archive Retrieval I/O} 
\label{sec:num;results}

In the previous section, we studied the static resilience for the various schemes - and demonstrated that, for a given storage overhead, all the strategies achieve the same effective fault-tolerance for the colocated placement scenario, which is the best case, and hence practical. We next focus on the disk I/Os involved while retrieving a versioned archive. The actual savings when using SEC depend on the actual sparsity of the differences across versions, and hence we obtain these results numerically for different example workloads.

\subsection{Non-systematic and Systematic SEC}
\label{sec:simresfizedgamma}
We consider the object from Section \ref{running_ex1} where we chose a specific case of $\mathbf{z}_{2}$ being $1$-sparse. However, in general, the sparsity level of $\mathbf{z}_{2}$ is a random variable over the support $\{1, 2, 3\}$. Since $k = 3$, the sparsity can be exploited only when $\gamma_{2} = 1$. Henceforth, we denote $\gamma_{2}$ by $\gamma$. We now discuss the number of options for the systematic and non-systematic codes to retrieve the $1$-sparse object with just 2 I/O reads. For the non-systematic code, since any $2 \times 3$ submatrix of  $\mathbf{G}_{N}$ satisfies the Criterion 2, there are a total of 15 such matrices. However, for the systematic code, only $3$ submatrices of $\mathbf{G}_{S}$ satisfy Criterion 2. 

\begin{figure}
\begin{center}
\includegraphics[width=3in]{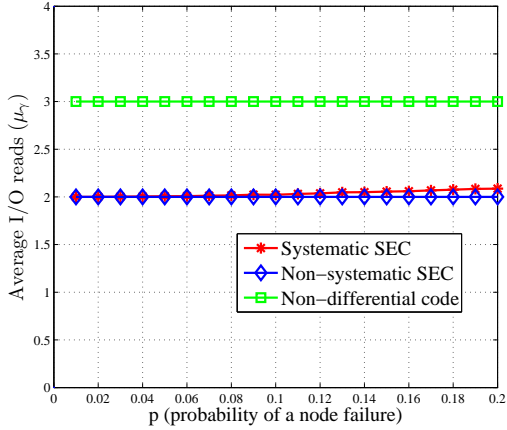}
\vspace{-0.2cm}
\caption{Average I/O reads $\mu_{\gamma}$ given in \eqref{avg_i_o_metric} for $\gamma = 1$ to retrieve the $1$-sparse object $\mathbf{z}_{2}$ for different erasure codes. These results are obtained for the parameters $n = 6$ and $k = 3$ in Section \ref{running_ex1}.}
\label{fig_avg_reads1}
\end{center}
\end{figure}

We next discuss the implications of having reduced number of submatrices satisfying Criterion 2 on the (average) I/O for retrieving the $\gamma$-sparse object $\mathbf{z}_{2}$. The following approach is needed only for systematic SEC, since for the non-systematic SEC, every pattern of $k$ or more live nodes has a subset of size $2\gamma$ that can recover $\mathbf{z}_{2}$, and hence $2\gamma$ I/O reads are guaranteed. We randomly generate a large ensemble of failure patterns of nodes by assuming that each node fails independently with probability $p$. We traverse through every failure pattern to identify if at least $k$ live nodes remain, in order to recover the object. If $k$ or more nodes are alive, we find the possibility of retrieving $\mathbf{z}_{2}$ with just $2\gamma$ I/O reads by checking Criterion 2 on the submatrices corresponding to live nodes. Accordingly, we retrieve the entire object with just $2\gamma$ I/O reads, otherwise, it is recovered using $k$ reads. Provided that a failure pattern leaves $k$  or more live nodes, we compute the percentage of cases when only $2\gamma_{2}$ reads are sufficient (denoted by $p_{2\gamma}$) and that of when $k$ I/O reads are needed (denoted by $p_{k}$). Then the average number of I/O reads is computed as 
\begin{equation}
\label{avg_i_o_metric}
\mu_{\gamma} = p_{2\gamma}2\gamma + p_{k}k. 
\end{equation}
For the parameters in Section \ref{running_ex1}, we have $\mu_{1} = 2p_{2} + 3p_{3}$. In Fig. \ref{fig_avg_reads1}, we plot $\mu_{1}$ for different values of $p$ from $0.01$ to $0.2$. The plot shows that when $p$ is small, systematic SEC recovers $\mathbf{z}_{2}$ with just 2 I/O reads. However, as $p$ increases, a non-negligible number of error patterns occurs for which no subset of live nodes with cardinality 2 can recover $\mathbf{z}_{2}$. The two other schemes are also shown: (i) the lower one (with constant reads of 2) corresponds to the non-systematic SEC as every pattern of $k$ or more live nodes has a subset of size 2 that can recover $\mathbf{z}_{2}$, and (ii) the top one (with constant reads of 3) corresponds to non-differential encoding where sparsity cannot be exploited.

A similar experiment is repeated with parameters $n = 10, k = 5$ and $L = 2$. The first version $\mathbf{x}_{1}$ is fully encoded and the second version is encoded with the SEC schemes. Since $k = 5$, we apply the average I/O reads study to $\gamma = 1$ and $2$. Provided that failure patterns are such that $k$ or more nodes are alive, the average number of I/O reads $\mu_{\gamma}$ (given in \eqref{avg_i_o_metric}) to retrieve $\mathbf{z}_{2}$ are provided in Fig. \ref{fig_avg_reads2} for (i) $\gamma = 1$ and (ii) $\gamma = 2$. The plots show that for $\gamma = 1$, systematic SEC retrieves $\mathbf{z}_{2}$ using $2$ I/O reads almost always for values of $p$ till $p = 0.2$. However, for $\gamma = 2$, there is marginal increase in the values of $\mu_{\gamma}$ for higher values of $p$ till $p = 0.2$.

\begin{figure}
\begin{center}
\includegraphics[width=3.5in]{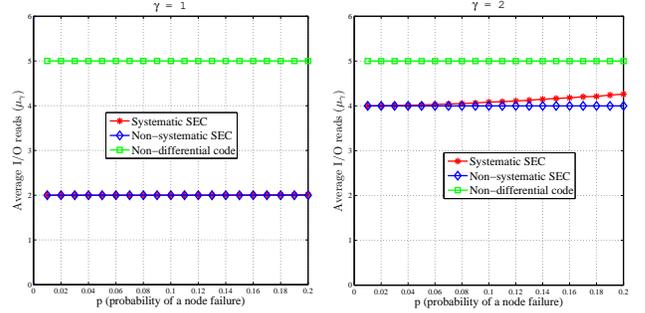}
\vspace{-0.8cm}
\caption{Average I/O reads given in \eqref{avg_i_o_metric} to retrieve the sparse object $\mathbf{z}_{2}$ for different erasure codes. The results are obtained for the parameters $n = 10$ and $k = 5$. The plots on the left and right correspond to retrieving $\mathbf{z}_{2}$ which is $1$-sparse and $2$-sparse, respectively.}
\label{fig_avg_reads2}
\end{center}
\end{figure}

In conclusion (i) both variants of SEC outperform the naive solution of encoding individual versions, and (ii) while non-systematic SEC consistently performs better than systematic SEC, the differences are marginal for practical settings (where $p$) is not very high.

\subsection{Expected I/O savings with two versions}

Previously, we saw in Section \ref{sec2} that for the example from Section \ref{running_ex1}, SEC reduces the I/O reads for joint retrieval of $\mathbf{x}_{1}$ and $\mathbf{x}_{2}$ from $6$ to $5$ when $\mathbf{z}_{2}$ is $1$-sparse. We further studied settings with fixed $\gamma$ values in Section \ref{sec:simresfizedgamma} above. However, in general, the sparsity level of $\mathbf{z}_{2}$ can take any value over $\{1, 2, 3\}$. We now present numerical results on the expected number of I/O reads when $\gamma_{2}$ is random (denoted by the random variable $\Gamma_{2}$). 

For the SEC schemes, the number of I/O reads to access both the versions are $5, 6,$ and $6$ when $\mathbf{z}_{2}$ is $1$-, $2$- and $3$-sparse, respectively. Since the average number of I/O reads depends on the underlying probability mass function (PMF) on the sparsity level, we study the advantages of the proposed method by testing different PMFs that reflect different difference sparsity behaviors, in the absence of standard workloads (see e.g. \cite{dedup}) . Henceforth, we use $\Gamma$ to denote the random variable $\Gamma_{2}$ and $\gamma$ to denote its realization $\gamma_{2}$. The PMF on $\Gamma$ is denoted by $\mbox{P}_{\Gamma}(\gamma)$ for $\gamma \in \{1, 2, 3\}$.

\textbf{PMFs on sparsity.} We apply the finite support versions of the exponential distribution in parameter $\alpha > 0$ given by 
\begin{equation}
\label{normalised_exponential}
\mbox{P}_{\Gamma}(\gamma) = c e^{-\alpha \gamma}, \mbox{ for } \gamma = 1, 2, 3, 
\end{equation}
where the constant $c$ is chosen such that $\sum_{\gamma = 1}^{k} \mbox{P}_{\Gamma}(\gamma) = 1$, and referred to as \emph{truncated exponential} PMF. Likewise, a \emph{truncated Poisson} PMF with parameter $\lambda$, given by 
\begin{equation}
\label{normalised_poisson}\
\mbox{P}_{\Gamma}(\gamma) = c \frac{\lambda^{\gamma} e^{-\lambda}}{\gamma!}, \mbox{ for } \gamma = 1, 2, 3, 
\end{equation}
are also considered, where $c$ is such that $\sum_{\gamma = 1}^{k} \mbox{P}_{\Gamma}(\gamma) = 1$. These PMFs are specifically picked to study the reduction in I/O reads for two extreme scenarios: (i) the family of exponential PMFs provides thick concentration towards smaller value of $\Gamma$, whereas (ii) the family of Poisson PMFs provides thick concentration towards larger value of $\Gamma$. Thus they facilitate the study of both the best-case and worst-case scenarios for SEC. In Figures \ref{fig_exp_poiss_pmf}, we plot the PMFs in \eqref{normalised_exponential} and \eqref{normalised_poisson}, respectively for different parameters. 

\begin{figure}
\begin{center}
\includegraphics[width=3.3in]{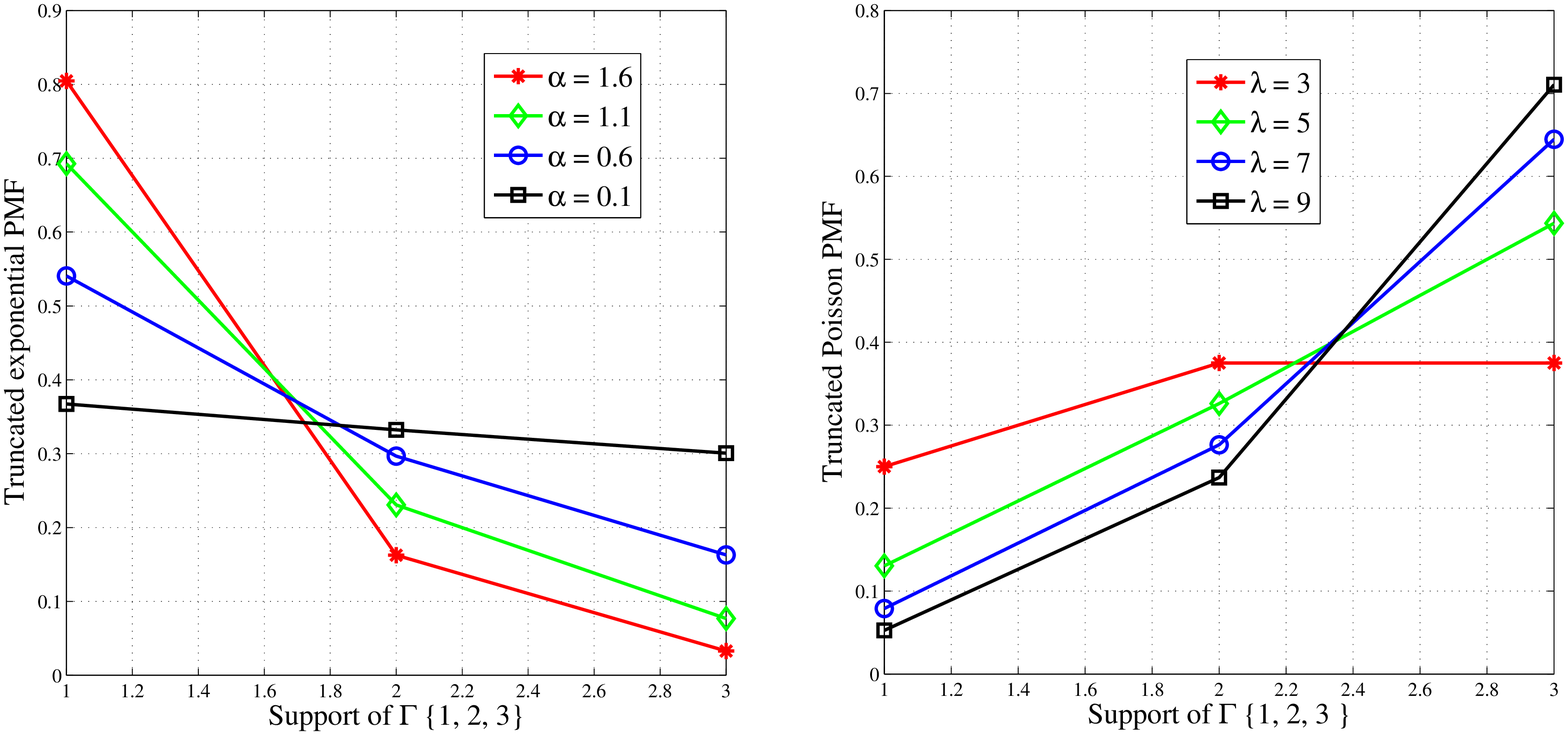}
\vspace{-0.5cm}
\caption{Truncated exponential PMFs (with parameter $\alpha$) and Poisson PMFs (with parameter $\lambda$) on $\Gamma$ for $k = 3$. The x-axis represents the support $\{1, 2, 3\}$ of $\Gamma$.}
\label{fig_exp_poiss_pmf}
\end{center}
\end{figure}

 
\textbf{Average I/O reads to retrieve $\mathbf{x}_{1}$ and $\mathbf{x}_{2}$.} For a given $\mbox{P}_{\Gamma}(\gamma)$, the average number of I/O reads for accessing the first two versions are given by
$$\mathbb{E}[\eta] = k + \sum_{\gamma = 1}^{k} \mbox{P}_{\Gamma}(\gamma) \mbox{min}(2\gamma, k).$$ In Figure \ref{fig_io_exp__poiss_adv}, we plot the average percentage reduction in the I/O reads when compared to the non-differential setup as $\frac{2k - \mathbb{E}[\eta]}{2k} \times 100$
where $2k$ is the total number of I/O reads for the non-differential scheme. The plots show a significant reduction in the I/O reads when the distribution is skewed towards smaller $\gamma$. However, as expected, the reduction is marginal otherwise.

\begin{figure}
\begin{center}
\includegraphics[width=3.3in]{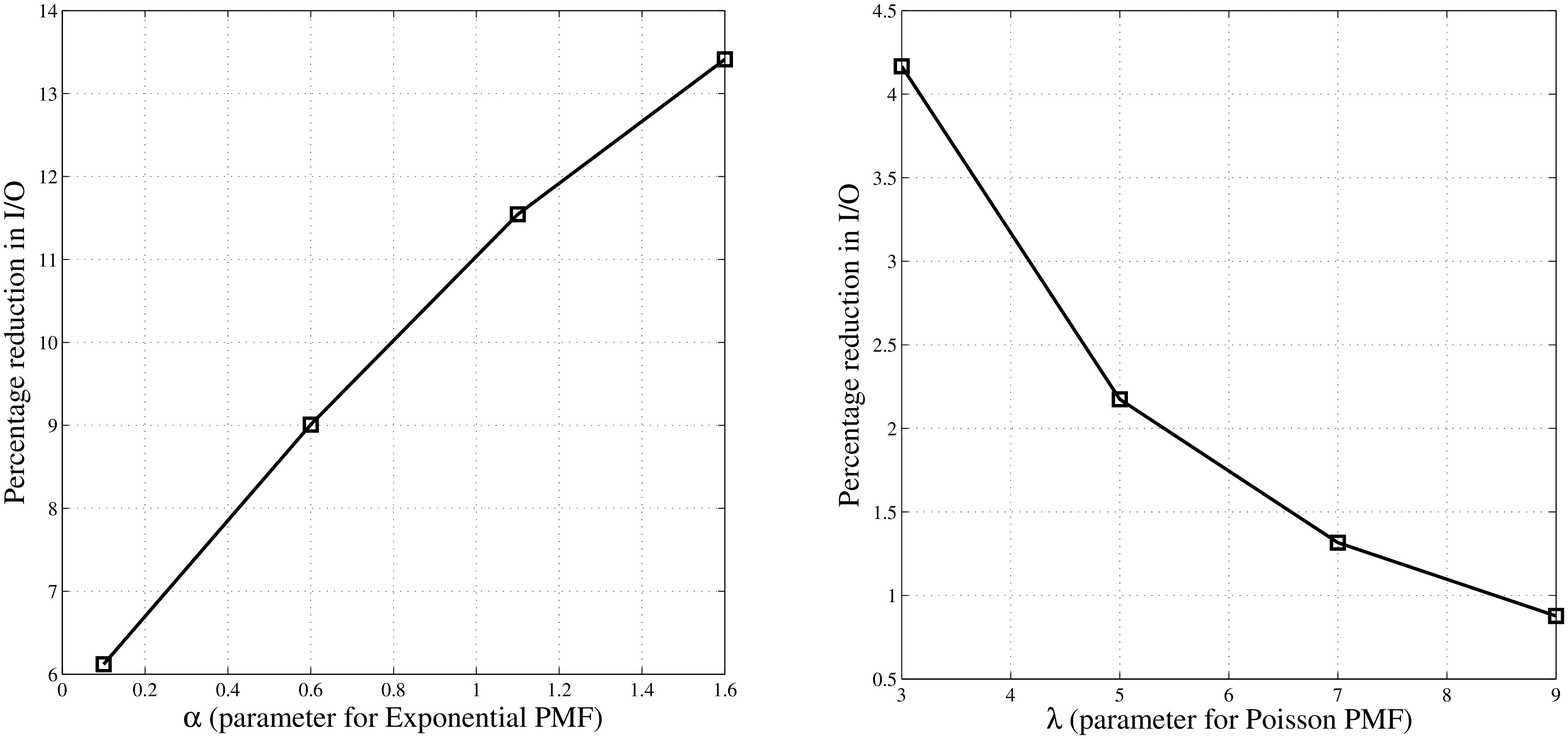}
\vspace{-0.5cm}
\caption{Average percentage reduction in the I/O reads to access $\mathbf{x}_{1}$ and $\mathbf{x}_{2}$ for PMFs in Fig. \ref{fig_exp_poiss_pmf}. Here $\alpha$ and $\lambda$ are the parameters of the truncated Exponential PMF and truncated Poisson PMFs in \eqref{normalised_exponential} and \eqref{normalised_poisson}, respectively. The results are for $n = 6$ and $k = 3$.}
\label{fig_io_exp__poiss_adv}
\end{center}
\end{figure}


\textbf{Average I/O reads to retrieve $\mathbf{x}_{2}$ alone.} The average number of I/O reads to retrieve the 2nd version alone using the basic SEC is $\mathbb{E}[\eta(\mathbf{x}_{2})] = \mathbb{E}[\eta(\mathbf{x}_{1}, \mathbf{x}_{2})]$ since the delta has to applied over the first version. However, for the optimized method, the average I/O reads is $\mathbb{E}[\eta(\mathbf{x}_{2})] = \sum_{\gamma = 1}^{k} \mbox{P}_{\Gamma}(\gamma) t(\gamma)$ where $t(\gamma) = k$ when $\gamma \geq \frac{k}{2}$, and $t(\gamma) = k  + 2\gamma$, otherwise. Compared to non-differential coding, the average percentage increase in the I/O reads for fetching the 2nd version for both the basic and the optimized methods is computed as $\frac{\mathbb{E}[\eta(\mathbf{x}_{2})] - k}{k} \times 100$.
In Fig. \ref{fig_io_x2_exp_poiss_pmf} we present the results corresponding to the PMFs in Fig. \ref{fig_exp_poiss_pmf}. It shows that the optimized SEC reduces the excess number of I/O reads for the 2nd version.  Though the optimized SEC reduces the excess I/O, this additional I/O reads for $\mathbf{x}_{2}$ is due to differential encoding that reduces the I/O for accessing both $\mathbf{x}_{1}$ and $\mathbf{x}_{2}$. One possible direction to reduce the I/O for the latest version is to employ reverse SEC (as pointed in Section \ref{sec2_subsec1}).

\begin{figure}
\begin{center}
\includegraphics[width=3.3in]{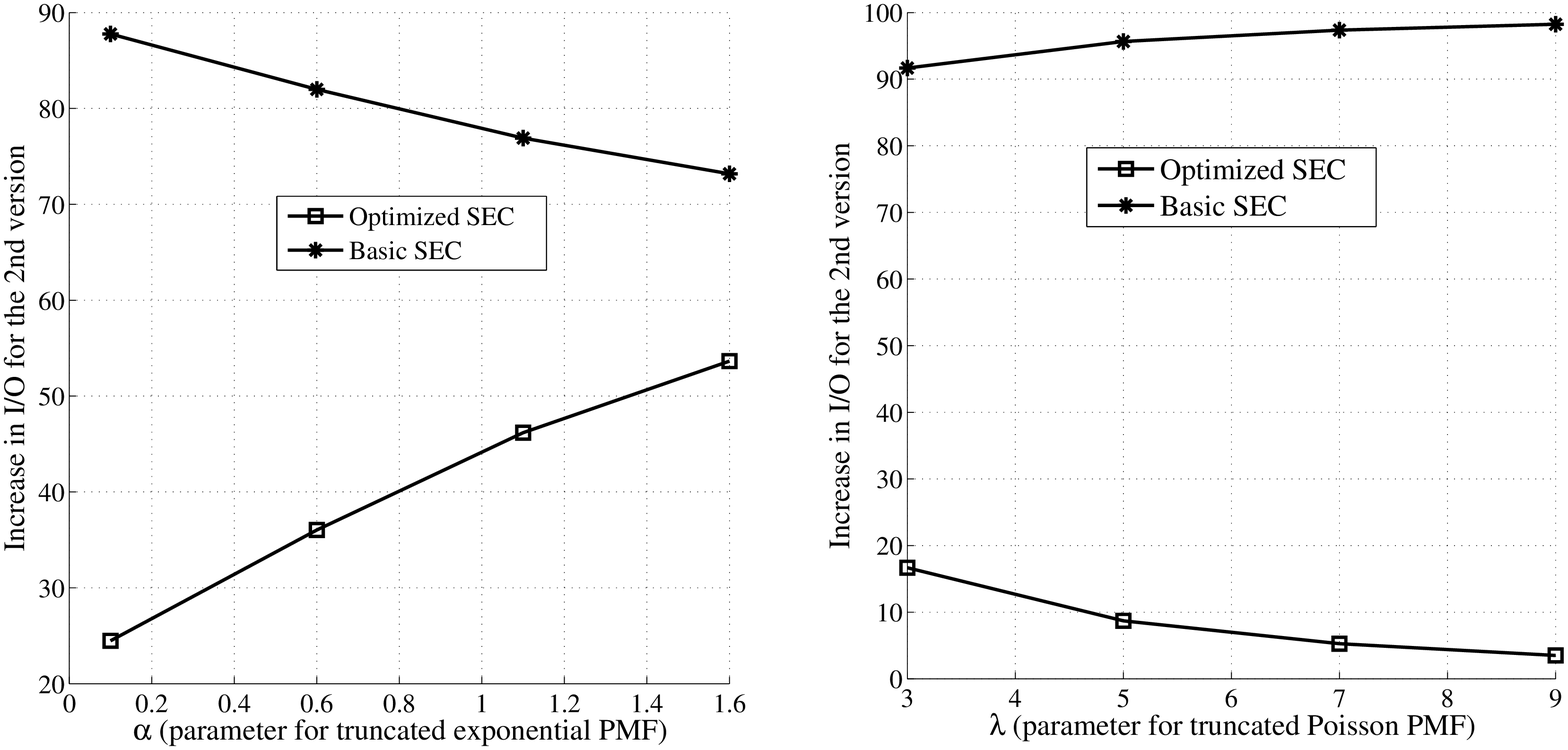}
\vspace{-0.5cm}
\caption{Average percentage increase in the I/O reads to access $\mathbf{x}_{2}$ for PMFs in Fig. \ref{fig_exp_poiss_pmf}. Results for both basic and optimized SEC methods are presented. Here $\alpha$ and $\lambda$ are the parameters of the truncated Exponential PMF and truncated Poisson PMFs in \eqref{normalised_exponential} and \eqref{normalised_poisson}, respectively.}
\label{fig_io_x2_exp_poiss_pmf}
\end{center}
\end{figure}


\subsection{An example system with $L > 2$ versions} To study the trade-offs of using SEC when multiple versions are involved, we revisit the example in Section \ref{sec2:subsec:example}. $L = 5$ versions with sparsity levels of subsequent versions $\{\gamma_{j} ~|~ 2 \leq j \leq L\} = \{3, 8, 3, 6\}$ of an object of size $k = 10$ is stored using a $(20, 10)$ SEC. We plot the I/O numbers for the basic and the optimized SEC in Fig. \ref{fig_L_5_analysis}. The numbers are presented to retrieve both the individual versions ($l$-th version for $1 \leq l \leq 5$) as well as all the first $l$ versions. The plot shows 20\% saving in total I/O reads with respect to the non-differential scheme, for only slightly higher I/O for the optimized DEC. 

\begin{figure}
\begin{center}
\includegraphics[width=3in]{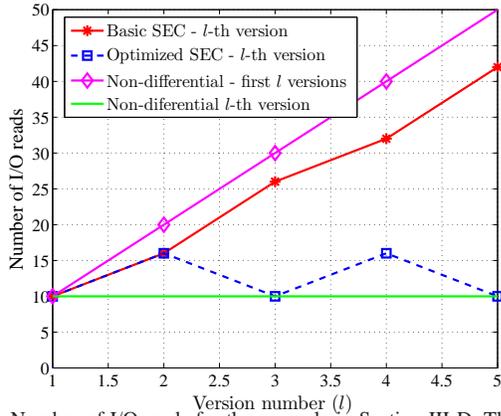}
\vspace{-0.5cm}
\caption{Number of I/O reads for the example in Section \ref{sec2:subsec:example}. The I/O reads to retrieve the $l$-th version and the first $l$-versions, $1 \leq l \leq 5$ are presented for different methods.}
\label{fig_L_5_analysis}
\end{center}
\end{figure}

\section{Conclusions}

In this paper we propose a framework - Sparsity Exploiting Coding (SEC) - for archiving versioned data using storage efficient erasure coding, where the individual versions are not coded in isolation, but instead the differences across subsequent versions are coded. The sparsity in the delta information is exploited for better I/O performance when the archive of versioned data is read back. We identify Cauchy matrix based MDS codes as one candidate which satisfies the requirements laid out in our framework to be able to opportunistically exploit the sparsity, and discuss two variants, a systematic and a non-systematic one. We demonstrate that, in doing so, for a given choice of storage overhead, and for a practical redundancy placement strategy (where encoded pieces pertaining to all the versions are colocated) there is no compromise in the system's resilience. Analysis and numerical/simulation experiments confirm the effectiveness of SEC. 

Our study shows that the non-systematic SEC provides better resilience for deltas corresponding to individual intermediate versions, even though the whole archive's resilience is constricted by the resilience of the first (or last) version of the data, which is coded as it is. This suggests room for reducing the storage overhead of the intermediate versions in the non-systematic SEC, possibly using puncturing techniques, which we will pursue in immediate future.

\bibliographystyle{abbrv}
\bibliography{IPDPS1515}

\begin{thebibliography}{10}

\bibitem{ECD}
K.~S. Esmaili, A.~Chiniah, and A.~Datta.
\newblock Efficient updates in cross-object erasure-coded storage systems.
\newblock In {\em IEEE International Conference on Big Data}, 2013.

\bibitem{google}
D.~Ford, F.~Labelle, F.~I. Popovici, M.~Stokely, V.-A. Truong, L.~Barroso,
  C.~Grimes, and S.~Quinlan.
\newblock Availability in globally distributed storage systems.
\newblock In {\em The 9th USENIX conference on Operating Systems Designand
  Implementation (OSDI)}, 2010.

\bibitem{HPZV}
S.~Han, H.-T. Pai, R.~Zheng, and P.~K. Varshney.
\newblock Update-efficient regenerating codes with minimum per-node storage.
\newblock In {\em Proceedings of the Int. Symp. Inf. Theory}, 2013.

\bibitem{LaF}
J.~Lacan and J.~Fimes.
\newblock A construction of matrices with no singular square submatrices.
\newblock In {\em International Conference on Finite Fields and Applications},
  2003.

\bibitem{MWC}
A.~Mazumdar, G.~W. Wornell, and V.~Chandar.
\newblock Update efficient codes for error correction.
\newblock In {\em Proceedings of the Int. Symp. Inf. Theory}, 2012.

\bibitem{OgD2}
F.~Oggier and A.~Datta.
\newblock {\em Coding Techniques for Repairability in Networked Distributed
  Storage Systems}.
\newblock Foundations and Trends in Communications and Information Theory, Now
  Publishers, 2013.

\bibitem{RVBS}
A.~Rawat, S.~Vishwanath, A.~Bhowmick, and E.~Soljanin.
\newblock Update efficient codes for distributed storage.
\newblock In {\em Proceedings of the Int. Symp. Inf. Theory}, 2011.

\bibitem{svn}
"SVN".
\newblock \url{http://subversion.apache.org/}.

\bibitem{hadoop}
A.~Thusoo, Z.~Shao, S.~Anthony, D.~Borthakur, N.~Jain, J.~S. Sarma, R.~Murthy,
  and H.~Liu.
\newblock Data warehousing and analytics infrastructure at facebook.
\newblock In {\em Proceedings of the 2010 ACM SIGMOD International Conference
  on Management of data, ser. SIGMOD ’10}, 2010.

\bibitem{dedup}
V.Tarasov, A.~Mudrankit, W.~Buik, P.~Shilane, G.~Kuenning, and E.~Zadok.
\newblock Generating realistic datasets for deduplication analysis.
\newblock In {\em Proceedings of the 2012 USENIX conference on Annual Technical
  Conference}, 2012.

\bibitem{ZP}
F.~Zhang and H.~D. Pfister.
\newblock Compressed sensing and linear codes over real numbers.
\newblock In {\em Information Theory and Applications Workshop (ITA)}, 2008.

\end{thebibliography}

\end{document}